\documentclass[conference, 10pt, a4paper]{IEEEtran}
\usepackage{acronym}
\acrodef{3GPP}{3rd generation partnership project}
\acrodef{3D}{three-dimensional}
\acrodef{2D}{two-dimensional}
\acrodef{5G}{fifth generation}
 \acrodef{ABS}{almost blank subframe}
 \acrodef{AES}{advanced encryption standard}
\acrodef{APE}{absolute percentage error}
\acrodef{AR}{augmented reality}
    \acrodef{BS}{base station}
		\acrodef{BA}{bundle adjustment}
    \acrodef{CDF}{cumulative distribution function}
		\acrodef{CEVP}{constrained eigenvalue problem}
    \acrodef{CSI}{channel state information}
    \acrodef{CQI}{channel quality indicator}
		\acrodef{CNN}{convolutional neural network}
		\acrodef{CSP}{communication service provider}
\acrodef{DL}{downlink}
\acrodef{DoF}{degree of freedom}
\acrodef{DNN}{deep neural network}
\acrodef{DRL}{deep reinforcement learning}
\acrodef{DUDe}{downlink and uplink decoupling}
\acrodef{DDPG}{Deep Deterministic Policy Gradient}

\acrodef{eICIC}{enhanced intercell interference coordination}
\acrodef{ESD}{energy spectral density}
\acrodef{ECDSA}{elliptic curve digital signature algorithm}
\acrodef{XR}{extended reality}
\acrodef{FDD}{frequency division duplex}
    \acrodef{FDMA}{frequency division multiple access}
		\acrodef{fps}{frame per second}
   \acrodef{GP}{Gaussian process}
    \acrodef{GPS}{global positioning system}
		\acrodef{GUI}{graphical user interface}
\acrodef{HetNet}{heterogeneous network}
\acrodef{HOG}{histogram of oriented gradients}
    \acrodef{ICI}{inter-cell interference}
		\acrodef{IMI}{inter-mode interference}
		\acrodef{IMU}{inertial measurement unit}
		\acrodef{IoT}{internet of things}
\acrodef{KL}{Kullback-Leibler}
\acrodef{LTE}{long term evolution}

\acrodef{mAP}{mean average precision}
\acrodef{MAC}{media access control}
\acrodef{MAPE}{mean absolute percentage error}
\acrodef{MARL}{multi-agent reinforcement learning}
\acrodef{MADRL}{Multi-agent deep reinforcement learning}
\acrodef{MDP}{Markov decision process}
\acrodef{MMDP}{multi-agent Markov Decision Process}
\acrodef{MIMO}{multiple-input and multiple-output}
\acrodef{MRU}{minimum resource unit}
\acrodef{mmWave}{millimeter wave}
\acrodef{MLP}{multi-layer perception}
\acrodef{ML}{machine learning}
\acrodef{MILP}{mixed integer linear programming}

\acrodef{NLES}{nonlinear equation system}
\acrodef{NSM}{Network Slicing Management}
\acrodef{NLP}{natural language processing}
   \acrodef{OFDM}{orthogonal frequency division multiplexing}
	\acrodef{ORB}{oriented FAST and rotated BRIEF}
    \acrodef{PDF}{probability density function}
    \acrodef{PHY}{physical layer}
		\acrodef{PSD}{power spectral density}
    \acrodef{PRB}{physical resource block}
	\acrodef{PRBs}{physical resource blocks}
   \acrodef{QoE}{quality of experience}
    \acrodef{QoS}{quality of service}
    \acrodef{RAN}{radio access network}
		\acrodef{RANSAC}{random sample consensus}
		\acrodef{RB}{resource block}
		\acrodef{RBS}{removal of bottleneck services}
		\acrodef{RL}{reinforcement learning}
		\acrodef{R-FCN}{region-based fully convolutional networks}
		\acrodef{RMDI}{resource muting for dominant interferer}
		\acrodef{RMSD}{root mean square distance}
		\acrodef{ROI}{region of interest}
		\acrodef{RPN}{region proposal network}
    \acrodef{RRM}{radio resource management}
		\acrodef{RU}{resource unit}
		\acrodef{RX}{receiver}
 \acrodef{SAFP}{successive approximation of fixed point}
 \acrodef{SE}{secure element}
    \acrodef{SDN}{software defined network}
    \acrodef{SNR}{signal-to-noise ratio}
    \acrodef{SINR}{signal-to-interference-plus-noise ratio}
\acrodef{SIR}{signal-to-interference ratio}
\acrodef{SIF}{standard interference function}
\acrodef{SLAM}{simultaneous localization and mapping}
\acrodef{SLAMORE}{Simultaneous Localization and Mapping with Object REcognition}
\acrodef{SoC}{system-on-chip}
\acrodef{SONs}{self organizing networks}
\acrodef{SSD}{single-shot multibox detector}
\acrodef{SPI}{serial peripheral interface}
    \acrodef{SVM}{support vector machine}
		\acrodef{SVD}{singular value decomposition}
    \acrodef{TCP}{transmission control protocol}
		\acrodef{TLS}{transport layer security}
		\acrodef{TL}{transfer learning}
		\acrodef{TDD}{time division duplex}
	\acrodef{TD3}{Twin Delayed Deep Deterministic policy gradient}
    \acrodef{TDMA}{time division multiple access}
		\acrodef{TTI}{transmission time interval}
		\acrodef{TX}{transmitter}
		\acrodef{UART}{universal asynchronous receiver-transmitter}
		\acrodef{UDP}{user datagram protocol}
		\acrodef{UE}{user equipment}
		\acrodef{UI}{user interface}
		\acrodef{UL}{uplink}
		\acrodef{UAV}{unmanned aerial vehicle}
		\acrodef{VO}{visual odometry}
		\acrodef{V2X}{vehicle-to-everything}
		\acrodef{VAE}{variational auto-encoder}
    \acrodef{WLAN}{wireless local area network}
\acrodef{YOLO}{you only look once}
\usepackage[left=1.5cm,right=1.5cm,top=1.4cm,bottom=2.7cm]{geometry}
\usepackage[font=small,skip=2pt]{caption}
\usepackage{cite}
\usepackage{amsmath,amssymb,amsfonts}
\usepackage[utf8]{inputenc}
\usepackage[mathscr]{euscript}
\usepackage{algorithmic}
\usepackage{graphicx}
\usepackage{textcomp}
\usepackage{xcolor}

\usepackage{url}
\usepackage{bbm}
\DeclareMathAlphabet\mathbfcal{OMS}{cmsy}{b}{n}

\usepackage[english]{babel}
\usepackage{amsthm}

\newtheorem{theorem}{Theorem}[section]
\newtheorem{lemma}[theorem]{Lemma}
\newtheorem{problem}[theorem]{Problem}
\DeclareRobustCommand\optionalsec[1]{%
  \ifnum\pdfstrcmp{#1}{\thesection}=0\else#1.\fi
}

\DeclareMathOperator*\argmax{arg \, max}		
\DeclareMathOperator*\argmin{arg \, min}		

\newcommand{\field}[1]{\mathbb{#1}}

\newcommand{\set}[1]{\mathcal{#1}}
\newcommand{\Set}[1]{\mathbfcal{#1}}

\newcommand{\R}{{\field{R}}}   
\newcommand{\Ex}{{\field{E}}}
\newcommand{\NN}{{\field{N}}}

\newcommand{\ma}[1]{\boldsymbol{\mathbf{#1}}} 
\newcommand{\ve}[1]{\boldsymbol{\mathbf{#1}}} 

\newcommand{\vs}{\ve{s}}
\newcommand{\vx}{\ve{x}}
\newcommand{\vz}{\ve{z}}

\newcommand{\vm}{\ve{m}}

\newcommand{\vc}{\ve{c}}
\newcommand{\va}{\ve{a}}

\newcommand{\N}{{\set{N}}}

\newcommand{\K}{{\set{K}}}

\newcommand{\T}{{\set{T}}}
\newcommand{\M}{{\set{M}}}
\newcommand{\D}{{\set{D}}}
\newcommand{\Ss}{{\set{S}}}
\newcommand{\A}{{\set{A}}}
\newcommand{\X}{{\set{X}}}

\newcommand{\Z}{{\set{Z}}}

\newcommand{\bM}{\Set{M}}
\newcommand{\bD}{\Set{D}}
\newcommand{\bT}{\Set{T}}

\newcommand{\operator}[1]{\mathrm{#1}}
\newcommand{\diag}{\operator{diag}}

\begin{document}
\title{Network Slicing via Transfer Learning aided Distributed Deep Reinforcement Learning\vspace{-0.1in}}

\author{
\IEEEauthorblockN{Tianlun Hu\IEEEauthorrefmark{1}\IEEEauthorrefmark{3}, 
				  Qi Liao\IEEEauthorrefmark{1}, 
				  Qiang Liu\IEEEauthorrefmark{2}, 
				  and Georg Carle\IEEEauthorrefmark{3}}
\IEEEauthorblockA{ 
	\IEEEauthorrefmark{1}Nokia Bell Labs, Stuttgart, Germany\\
	\IEEEauthorrefmark{2}University of Nebraska Lincoln, United States\\
	\IEEEauthorrefmark{3}Technical University of Munich, Germany\\
Email: \IEEEauthorrefmark{1}\IEEEauthorrefmark{3}\url{tianlun.hu@nokia.com}, 
	   \IEEEauthorrefmark{1}\url{qi.liao@nokia-bell-labs.com}, 
	   \IEEEauthorrefmark{2}\url{qiang.liu@unl.edu},
	   \IEEEauthorrefmark{3}\url{carle@net.in.tum.de}}
	   \vspace{-0.4in}
}

\maketitle

\begin{abstract}
Deep reinforcement learning (DRL) has been increasingly employed to handle the dynamic and complex resource management in network slicing. The deployment of DRL policies in real networks, however, is complicated by heterogeneous cell conditions. In this paper, we propose a novel transfer learning (TL) aided multi-agent deep reinforcement learning (MADRL) approach with inter-agent similarity analysis for inter-cell inter-slice resource partitioning. First, we design a coordinated MADRL method with information sharing to intelligently partition resource to slices and manage inter-cell interference. Second, we propose an integrated TL method to transfer the learned DRL policies among different local agents for accelerating the policy deployment. The method is composed of a new domain and task similarity measurement approach and a new knowledge transfer approach, which resolves the problem of \emph{from whom to transfer} and \emph{how to transfer}. We evaluated the proposed solution with extensive simulations in a system-level simulator and show that our approach outperforms the state-of-the-art solutions in terms of performance, convergence speed and sample efficiency. Moreover, by applying TL, we achieve an additional gain over $27\%$ higher than the coordinated MADRL approach without TL.
\end{abstract}

\section{Introduction}

Network slicing is the key technique in 5G and beyond which enables network operators to support a variety of emerging network services and applications, e.g., autonomous driving, metaverse, and machine learning.
The virtual networks (aka. network slices) are dynamically created on the common network infrastructures, e.g., base stations, which are highly customized in different aspects to meet the diverse performance requirement of these applications and services.
As the ever-increasing network deployment, e.g., small cells, the traffic of slices and inter-cell interference in radio access networks become more dynamic and complex.
Conventional model-based solutions, e.g., linear programming or convex optimization, can hardly handle the ever-complicating resource management problem.

Recent advances in machine learning, especially \ac{DRL} \cite{ConstrainedRLNetSlicing}, \cite{DeepSlicingQiang}, has shown a promising capability to deal with the dynamic and high-dimensional networking problems.
The machine learning techniques, as model-free approaches, learn from historical interactions with the network, which require no prior knowledge, e.g., mathematical models.
Several works studied to formulate resource management problems as \ac{MDP}s, which are then solved by using \ac{DRL} to derive a centralized policy with global observations of the network.
As the network scale grows, the action and state space of the centralized problem increases exponentially, which challenges the convergence and sample efficiency of \ac{DRL}.
\ac{MADRL} \cite{Zhao2019DeepRL, Shao2021GraphAN} has been exploited to address this issue, which creates and trains multiple cooperative \ac{DRL} agents, where each \ac{DRL} agent focuses on an individual site or cell.
However, training all individual \ac{DRL} agents from scratch can still be costly and time-consuming, e.g., expensive queries with real networks, and unstable environments from the perspective of individual \ac{DRL} agents.

Recently, \ac{TL} \cite{pan2009survey} based methods have been increasingly studied to improve the sample efficiency and model reproducibility in the broad machine learning fields \cite{nguyen2021transfer, wang2021transfer,Parera20}.
The basic idea of \ac{TL} is to utilize prior knowledge from prelearned tasks to benefit the training process in new tasks.
For example, the resource partitioning policy of a cell can be transferred to another cell when they share similar network settings, e.g., bandwidth, transmit power, and traffic pattern.
Generally, there are several questions to be answered before using \ac{TL} methods, i.e., \emph{what to transfer}, \emph{from whom to transfer}, and \emph{how to transfer}. 
Existing \ac{TL} methods are mostly focused on supervised machine learning, e.g., computer vision and natural language processing \cite{zhuang2020comprehensive}, which provide limited insights on applying in \ac{DRL} tasks \cite{Taylor2009TransferLF,DBLP:journals/corr/abs-2009-07888, Nagib2021TransferLA, mai2021transfer}.
Therefore, it is imperative to study how \ac{TL} improves the performance of \ac{MADRL} in terms of sample efficiency and fine-tune costs, in the inter-cell resource partitioning problem.

In this paper, we proposed a novel \ac{TL} aided \ac{MADRL} approach with domain similarity analysis for inter-slice resource partitioning.
First, we design a coordinated \ac{MADRL} method for inter-cell resource partitioning problems in network slicing, where \ac{DRL} agents share local information with each other to mitigate inter-cell interference.
The objective of \ac{MADRL} is to maximize the satisfaction level of per-slice service requirements in terms of average user throughput and delay in each cell. 
Second, we design an integrated \ac{TL} method to transfer the learned DRL policies among different agents for accelerating the policy deployment, where the new method consists of two parts.
On the one hand, we propose a feature-based inter-agent similarity analysis approach, which measures the domain and task difference by extracting representative feature distributions in latent space.
On the other hand, we propose a new knowledge transfer approach with the combined model (policy) and instance transfer.
The main contributions of this paper are summarized as follows:
\begin{itemize}
	\item We design a coordinated \ac{MADRL} method for the inter-cell resource partitioning problem in network slicing.
 	\item We design a novel inter-agent similarity analysis approach, based on the features extracted by \ac{VAE} to evaluate both domain and task similarity between two reinforcement learning agents.
  	\item We design a new knowledge transfer approach that combines the model (policy) and instance transfer from the selected source agent to the target agent.
  	\item We evaluate the performance of the proposed solution with extensive simulations in a system-level simulator. The results show that, by applying \ac{TL}, we achieve an additional gain over $27\%$ higher than the coordinated \ac{MADRL} approach without \ac{TL}. Moreover, the performance gain achieved by \ac{TL} is more significant in the low-data regime.
\end{itemize}

\begin{figure}[t]
	\centering
\includegraphics[width=0.42\textwidth]{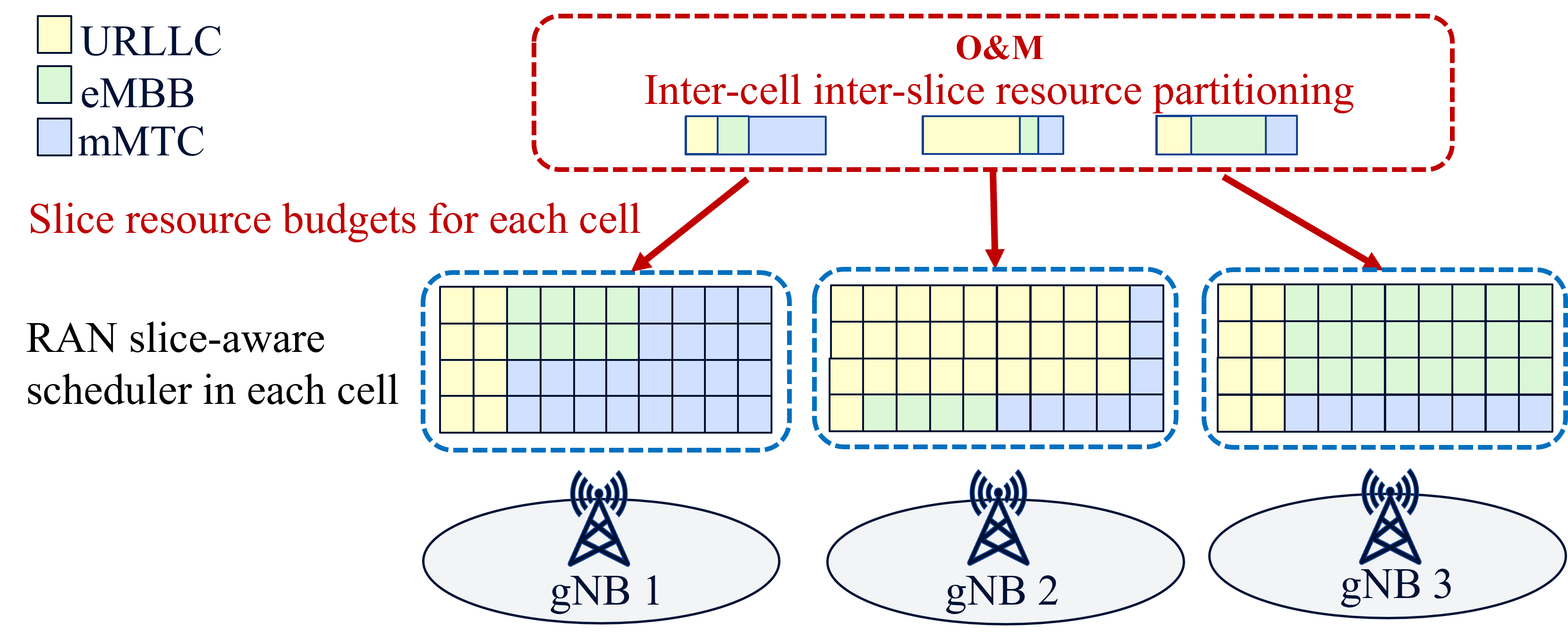}
	\vspace{0ex}
	\caption{Dynamic multi-cell slicing resource partitioning}
	\label{fig:RANSlicing}
	\vspace{0ex}
\end{figure}
\vspace{-0.04in}
\section{System Model and Definitions}
\vspace{-0.04in}
We consider a network consisting of a set of cells $\K:=\{1, 2, \ldots, K\}$ and a set of slices $\N:=\{1, 2, \ldots, N\}$. Each slice $n\in\N$ has predefined average user throughput and delay requirements, denoted as $\phi_n^*$ and $d_n^*$ respectively. The network system runs on discrete time slots $t \in \NN_0$. As illustrated in Fig. \ref{fig:RANSlicing}, network operation and maintenance (O\&M) adapts the inter-slice resource partitioning for all cells to provide per-slice resource budgets to each cell periodically. Then, within each cell, the \ac{RAN} scheduler uses the provided resource budgets as constraints and performs resource scheduling and \ac{PRB} allocation. In this paper, we focus on the inter-cell inter-slice resource partitioning problem in network O\&M. 

Considering the diverse slice requirements and dynamic network conditions, we model the multi-cell resource partitioning system as a set of $K$ distributed \acp{MDP} $\bM:=\{\M_1,...,\M_K\}$, with $\M_k:=\{\Ss_k, \A_k, P_k(\cdot), r_k(\cdot), \gamma_k\}$ defined for each agent $k\in\K$ (with a slight abuse of notation, hereafter we use $k$ for cell and agent interchangeably). $\Ss_k$ and $\A_k$ denote the state space and action space respectively. $P_k(\cdot):\Ss_k\times\A_k\times\Ss_k\rightarrow[0, 1]$ is the transition probability over $\Ss_k$ and $\A_k$ for cell $k$. $r_k: \Ss_k\times\A_k\rightarrow\R$ is defined as the reward function which evaluates the network service of all slices in cell $k$ and $\gamma_k$ denotes the discount factor for cumulative reward calculation. 

At each time step $t$, agent $k$ collects state $\vs_k(t) \in \Ss_k$ and decides an action $\va_k(t)\in\A_k$ according to policy $\pi_k: \Ss_k\rightarrow\A_k$, which indicates the per-slice resource partitioning ratio $a_{k,n}\in[0, 1]$ for $n\in\N$ while aligning with inter-slice resource constraints. Thus, the local action space $\A_k$ yields

\begin{equation}
	\A_k := \left\{\va_k\bigg|a_{k,n}\in[0,1], \forall n\in\N; \sum_{n=1}^{N} a_{k,n} = 1\right\}.
	\label{eqn:local_actionspace}
\end{equation}

For each cell $k\in\K$, our objective is to maximize the minimum service satisfaction level in terms of average user throughput and delay $(\phi_n^*, d_n^*)$ over all slices. Thus, for each agent $k$, we define the local reward function based on the observed per-slice average user throughput $\phi_{k,n}(t)$ and delay $d_{k,n}(t)$ at time $t$ as
\begin{equation}
	r_k(t) := \min_{n\in\N} \min\left\{\frac{\phi_{k,n}(t)}{\phi_{k,n}^\ast}, \frac{d_{k,n}^\ast}{d_{k,n}(t)}, 1\right\}.
	\label{eqn:local_rewardfuntion}
\end{equation}

The reward formulation drops below $1$ when the actual average
throughput or delay of any slices fails to fulfill the requirements. Note that the reward is upper bounded by $1$ even if all slices achieve better performances than the requirements, to achieve more efficient resource utilization. The second item in \eqref{eqn:local_rewardfuntion} is {\bf inversely proportional} to the actual delay, namely, if the delay is longer than required this term is lower than $1$.
\vspace{-0.05in}
\section{Problem Formulation}\label{sec:problem}
\vspace{-0.05in}
{\bf The Reinforcement Learning Problem}: The problem is to find a policy $\pi_k: \Ss_k\rightarrow\A_k$ for each $k\in\K$ that predicts optimal inter-slice resource partitioning $\va_k(t)\in\A_k$ base on the local state $\vs_k(t)\in\Ss_k$ dynamically, to maximize the expectation of the cumulative discounted reward $r_k(t)$ defined in \eqref{eqn:local_rewardfuntion}, in a finite time horizon $T$. The problem is given by:
\begin{equation}
\max_{\pi_k; \va_k(t)\in\A_k} \Ex_{\pi_k} \left[\sum_{t=0}^{T} \gamma_k^t r_k\big(\vs_k(t), \va_k(t) \big)\right], \ \forall k\in\K,
	\label{eqn:DRL_problem}
\end{equation}
where $\A_k$ is defined in \eqref{eqn:local_actionspace}.

In our previous work \cite{Hu2022InterCellReDRL}, we proposed a coordinated multi-agent \ac{DRL} approach to transform an \ac{MADRL} problem to the distributed \ac{DRL} problem similar to \eqref{eqn:DRL_problem}, where the extracted information from neighboring cells is included into the state observation to better capture the inter-agent dependency. However, training all local agents in parallel from scratch can be costly and time-consuming. Moreover, the trained models are sensitive to environment changes and the retraining cost can be high.

Thus, in this paper, we raise the following new questions:\\
\emph{Can we reuse the knowledge in a pretrained model? When is the knowledge transferable? And, most importantly, how to transfer the gained knowledge from one agent to another?}


{\bf The Transfer Learning Problem}: To tackle the transfer learning problem, let us first introduce two definitions \emph{domain} and \emph{task} in the context of reinforcement learning. 

A \emph{domain} $\D:=\{\Ss, P(\vs)\}$ consists of a state feature space $\Ss$ and its probability distribution $P(\vs)$, for $\vs\in\Ss$. A \emph{task} $\T:=\{\A, \pi(\cdot)\}$ consists of the action space $\A$ and a policy function $\pi: \Ss\to \A$.

Thus, our inter-agent transfer learning problem is to find the optimal source agent among a set of pretrained agents, and transfer its knowledge (pretrained model and collected instances) to the target agent, such that problem \eqref{eqn:DRL_problem} can be solved in the target agent with fast convergence and limited amount of samples. In particular, the problem is defined in Problem \ref{prob:TL}.

\begin{problem}\label{prob:TL}
Given a set of pretrained source agents $\overline{\K}\subset\K$ with source domains $\bD^{(S)}:=\left\{\D_i^{(S)}: i\in\overline{\K}\right\}$ and pretrained tasks $\bT^{(S)}:=\left\{\T_i^{(S)}: i\in\overline{\K}\right\}$, also given any target agent $k\notin\overline{\K}$ with target domain $\D_k^{(T)}$ and untrained task $\T_k^{(T)}$, find the optimal source agent $i^{\ast}_k\in\overline{\K}$ for target agent $k$ to transfer knowledge such that 
\begin{align}
	\label{eqn:TL_problem}
	i^*_k := & \argmax\limits_{\substack{\pi_k|\pi^{(0)}_k = \Lambda\left(\pi^{(S)}_i\right);\\ i\in\overline{\K}}}
	 \ \Ex_{\pi_k} \left[\sum_{t=0}^{T} \gamma^t_k r_k\big(\vs_k(t), \va_k(t) \big)\right] \\
	& \mbox{s.t. } (\vs_k, \va_k) \in \Gamma\left(\D^{(S)}_{i}, \D_k^{(T)}, \A^{(S)}_{i}, \A_k^{(T)}\right), \nonumber
\end{align}
where $\Lambda\left(\pi^{(S)}_{i}\right)$ is the \emph{policy transfer} strategy which maps a pretrained source policy $\pi^{(S)}_{i}$ to the initial target policy $\pi^{(0)}_k$, while $\Gamma\left(\D^{(S)}_{i}, \D_k^{(T)}, \A^{(S)}_{i}, \A_k^{(T)}\right)$ is the \emph{instance transfer} strategy which selects the instances from the source agent, combines them with the experienced instances from the target agent, and saves them in the replay buffer for model training or fine-tuning in the target agent. More details about the transfer learning strategies will be given in Section \ref{subsec:TL_Approach}.
\end{problem}

\section{Proposed Solutions}
\label{subsec:Dift_DRL}
In this section, we first present a distributed \ac{MADRL} approach to solve the slicing resource partitioning problem in \eqref{eqn:DRL_problem}. Then, to solve problem \eqref{eqn:TL_problem} to find the optimal source agent, we propose a novel approach to inter-agent similarity analysis based on the extracted features using \ac{VAE}. Finally, for inter-agent transfer learning, we introduce transfer learning strategy which combines the model (policy) transfer and instance transfer.

\subsection{Coordinated \ac{MADRL} Approach}\label{subsec:MADRL}

As stated in \eqref{eqn:DRL_problem} , the distributed \ac{DRL} approach allows each agent to learn a local policy and makes its own decision on inter-slice resource partitioning based on local observation. Compared with the centralized \ac{DRL} approaches, distributed approaches reduce the state and action spaces and significantly accelerate the training progress. However, local observation alone cannot capture the inter-cell dependencies and provide sufficient information to achieve the globally optimal solution. Thus, we proposed in \cite{Hu2022InterCellReDRL} a distributed \ac{DRL} approach with inter-agent coordination which keeps the low model complexity while including the extracted information from neighboring cells to capture the inter-cell interference. We briefly summarize the coordinated distributed \ac{DRL} approach below, because we would like to focus on the main contribution, namely, the inter-agent transfer learning, in this paper. For more details, readers are referred to our previous work \cite{Hu2022InterCellReDRL}. 

Each local agent $k$ observes a local state $\vs'_k$, which contains the following network measurements:

\begin{itemize}
	\item Per-slice average user throughput $\{\phi_{k,n}:n\in\N\}$;
 	\item Per-slice network load $\{l_{k,n}:n\in\N\}$;
  	\item Per-slice number of users $\{u_{k,n}: n\in\N\}$.
\end{itemize}
Thus, with the above-defined three slice-specific features, the local state $\vs'_k$ has the dimension of $3N$. Additionally, to better capture the inter-cell dependencies and estimate the global network performance, we introduce an inter-agent coordination mechanism through network information sharing among agents. Let each agent $k$ broadcast a message $\vm_k$ to its neighboring group of agents, denoted by $\K_k$, which means, each agent $k$ receives a collection of messages $\overline{\vm}_k:=[\vm_i: i\in\K_k]\in\R^{Z^{(m)}}$. Instead of using all received messages in $\overline{\vm}_k$, we propose to to extract useful information $\vc_k\in\R^{Z^{(c)}}$ to remain the low model complexity. We aim to find an feature extractor $g: \R^{Z^{(m)}}\rightarrow\R^{Z^{(c)}}: \overline{\vm}_k\rightarrow\vc_k$, such that $Z^{(c)}\ll Z^{(m)}$. Then, we include the extracted features from the shared messages into the local state: $\vs_k:=[\vs'_k, \vc_k]$.

Knowing that the inter-agent dependencies are mainly caused by inter-cell interference based on cell load coupling \cite{Cavalcante2019ConnectionsBS}, we propose to let each cell $k$ share its per-slice load $l_{k,n}, \forall n\in\N$ to its neighboring cell. Then, we compute the extracted information $\vc_k$ as the average per-slice neighboring load. Namely, we define a deterministic feature extractor, given by:
\vspace{1ex}
\begin{equation}
	\label{neighbor_info}
	\begin{aligned}
	g_k: & \R^{N|\K_k|}\to\R^N : \left[l_{i, n}:n\in\N, i\in\K_k\right] \mapsto \vc_k(t)\\
	  \mbox{with } &  \vc_k(t):=\left[\frac{1}{|\K_k|}\sum_{i\in \K_k}l_{i,n}(t): n\in\N \right].
	 \end{aligned}
  \end{equation}

With the extended local state including the inter-agent shared information, we can use classical \ac{DRL} approaches, e.g., the actor-critic algorithms such as \ac{TD3} \cite{Fujimoto2018AddressingFA} to solve \eqref{eqn:DRL_problem}.

\subsection{Integrated \ac{TL} with Similarity Analysis}
\label{subsec:Domain_Sim}
The distributed \ac{DRL} approach introduced in Section \ref{subsec:MADRL} allows us to derive a set of pretrained local agents. Still, given a target cell $k$, e.g., a newly deployed cell, or an existing cell but with changed environment, more questions need to be answered: Can we transfer the prelearned knowledge from at least one of the pretrained agents? Which source cell provides the most transferable information? How to transfer the knowledge?

To solve the transfer learning problem in \eqref{eqn:TL_problem}, we develop a distance measure $\mathfrak{D}_{i,k}$ to quantify the inter-agent similarity between a source agent $i$ and a target agent $k$. We aim to transfer the knowledge from the source agent with the highest similarity (reflected by the lowest distance measure).

The ideal approach to analyze the domain and task similarity between two agents is to obtain their probability distributions of the state $P(\vs)$ and derive the conditional probability distribution $P(\va|\vs)$. However, the major challenge here lies in the limited samples in the target agent. Considering that the target agent is a newly deployed agent, there is no information available about its policy $P(\va|\vs)$, and $P(\vs)$ is very biased, because all samples are collected under the default configurations (i.e., constant actions).

Thus, we need to design a distance measure constrained by very limited and bias samples in the target agent, without any information about its policy $P(\va|\vs)$. Our idea is to derive and compare the {\bf joint state and reward distribution} under the same default action $\va'$, $P\left(\vs, r|\va=\va'\right)$, in both source and target agent. The rationale behind this is that, when applying the actor-critic-based \ac{DRL} architecture, the critic function estimates the Q value $Q_{\pi}(\va, \vs)$ based on action and state. Hence, the conditional probability $P(r|\vs,\va)$ should provide useful information of the policy. With $\va=\va'$, we can consider to estimate $P(r|\vs, \va=\va')$. To efficiently capture the information for both domain similarity (based on $P(\vs|\va=\va')$) and task/policy similarity (based on $P(r|\vs, \va=\va')$), we propose to estimate the joint probability $P(\vs, r|\va=\va') = P(r|\vs, \va=\va')P(\vs|\va=\va')$.

\begin{figure}[t]
	\centering
	\includegraphics[width=0.3\textwidth]{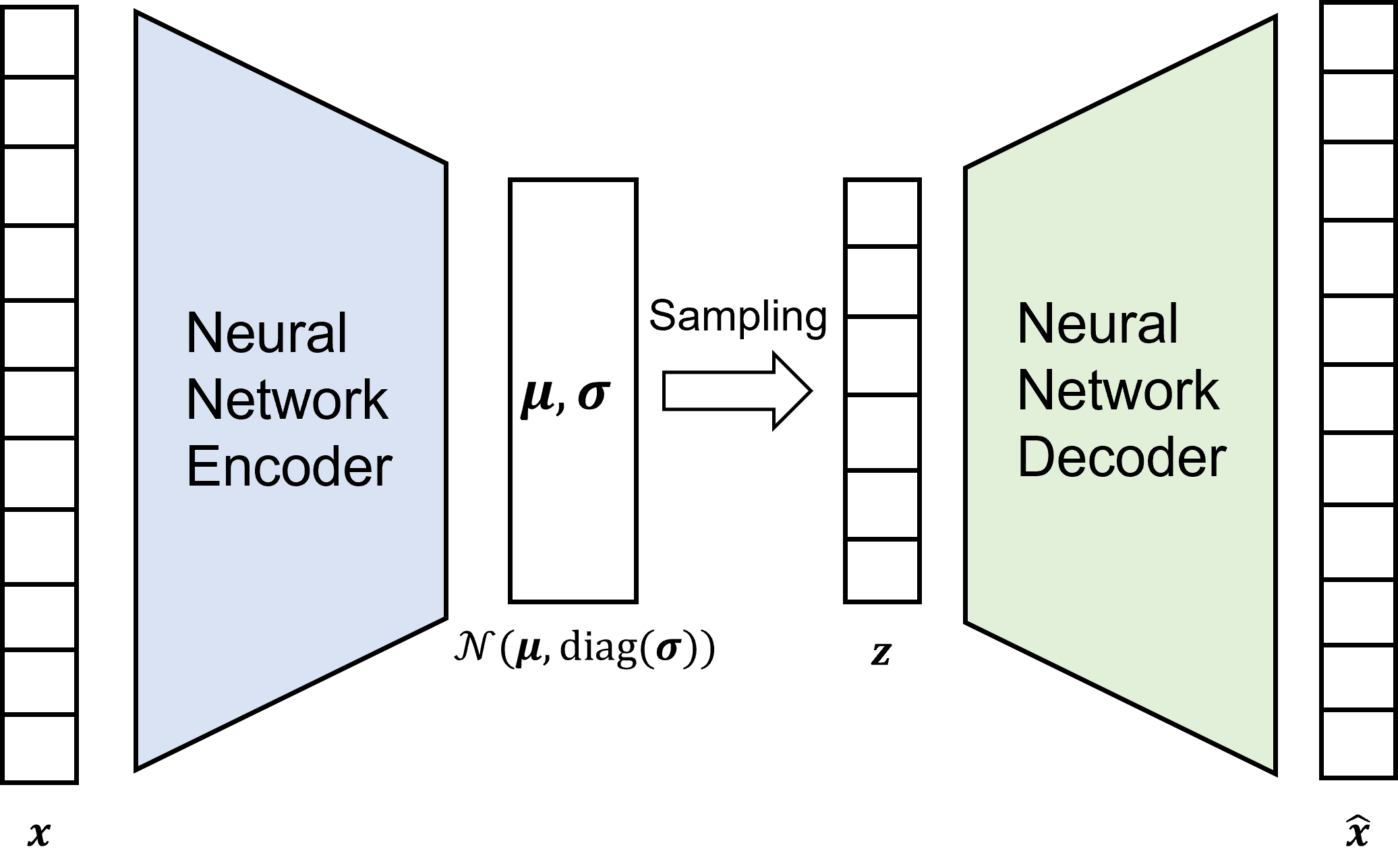}
	\vspace{-1ex}
	\caption{Variational autoencoder}
	\label{fig:VAE}
	\vspace{1ex}
\end{figure}

{\bf Sample collection}: To estimate the distance between $P(\vs, r|\va=\va')$ of both the source and target agents, we use all available samples from the target agent $k$ under the default action $\va'$, $\X_k = \{\left(\vs_k(n), r_k(n)\right)_{\va_k(n)=\va'}: n = 1, \dots, N_k\}$, and select a subset of the samples from the source agent $i$ with the same default action $\X_i = \{\left(\vs_i(n), r_i(n)\right)_{\va_i(n)=\va'}: n = 1, \dots, N_i\}$. Note that in this subsection we slightly abuse the notation by using $n$ as index of samples, and $N_k$ as number of samples with default action collected from agent $k$.

{\bf Feature extraction with \ac{VAE}}: To extract the representative features from the high-dimension vector $[\vs, r]$, we propose to apply \ac{VAE} \cite{Sohn2015LearningSO} to map the samples into a low dimensional latent space. As Fig. \ref{fig:VAE} illustrates, for each sample $\vx:=[\vs, r]\in\X$, the encoder of \ac{VAE} estimates an approximated distribution $P(\vz)$ in latent space $\Z$ as a multi-variate Gaussian distribution with $\N(\ve{\mu}, \diag(\ve{\sigma}))$, where $\diag$ denotes the diagonal matrix. The decoder samples a latent variable $\vz\in\Z$ from the approximated distribution $\vz\sim\N(\ve{\mu}, \diag(\ve{\sigma}))$ and outputs a reconstructed sample $\hat{\vx}$ by training on the following loss function:
\begin{align}
	\mathscr{L} := &  \|\vx-\hat{\vx}\|^2 + \nonumber\\
	& \alpha\cdot D_{KL}\left(\N(\ve{\mu}, \diag(\ve{\sigma}))\| \N(\ve{0},\diag(\ve{1}))\right), \label{eqn:VAE_loss}
\end{align}
where $\alpha$ is the weight factor and $D_{KL}$ denotes the \ac{KL} divergence.

{\bf Inter-agent similarity analysis}: Since \ac{VAE} does not directly provide the probability distribution function $P(\vx)$, we propose to utilize the extracted features in the latent space to evaluate the inter-agent similarity. Considering the limited amount of samples (only those under default action), we propose to train a general \ac{VAE} model based on the samples from all candidate source agents and the target agent, e.g., $\mathbf{\X} = \bigcup_{j\in\overline{\K}\cup\{k\}} \X_j$. The idea is to extract the latent features from samples from all relevant agents with a general encoder and to distinguish the agents within a common latent space.

Thus, for each sample $\vx_n\in\X$, we can derive its extracted features, i.e., the posterior distribution $P(\vz_n|\vx_n)=\N(\ve{\mu}_n, \diag(\ve{\sigma}_n))$. We denote the extracted latent space for agent $k$ by $\Z_k$.
Next, we can measure the inter-agent distance between an arbitrary source agent $i$ and target agent $k$ by calculating the \ac{KL} divergence based on the extracted latent variables from their collected samples:
\begin{align}
	\mathfrak{D}_{i,k}  & :=\frac{1}{N_iN_k}\cdot\nonumber\\
	\sum\limits_{\substack{(\ve{\mu}_n, \ve{\sigma}_n)\in\Z_i; \\ (\ve{\mu}_m, \ve{\sigma}_m)\in\Z_k}} & D_{KL}\left(\N(\ve{\mu}_{n}, \diag(\ve{\sigma}_{n}))\| \N(\ve{\mu}_{m}, \diag(\ve{\sigma}_{m}))\right). \label{domain_kl}
\end{align}
This requires to compute the \ac{KL} divergence of every pair of samples $(n,m)$ for $n\in\X_i$ and $m\in\X_k$, which could be computing intensive. 

Note that they are both Gaussian distributions, we can efficiently compute them with closed-form expression (as will be shown later in \eqref{eqn:KL_div}). 
Besides, from our experiment, we observed that $\ve{\sigma}_n\to \ve{0}$ for nearly all the collected samples $\vx_n\in\X$, i.e., their variances are extremely small (to the level below $10e-5$ from our observation).
Thus, for our problem, we can use a trick to evaluate the distance measure more efficiently based on the following lemma. 

\begin{lemma}
\label{lem:kl_mean_prop}
Given two multi-variate Gaussian distributions $p = \N(\ve{\mu}_{n}, \ma{\Sigma}_n)$ and $q = \N(\ve{\mu}_{m}, \ma{\Sigma}_m)$, where $\ve{\mu}_{n}, \ve{\mu}_{m}\in\R^L$, $\ma{\Sigma}_n=\ma{\Sigma}_m = \diag(\ve{\sigma})\in\R^{L\times L}$ and every entry of $\ve{\sigma}$ is equal to a small positive constant $\sigma\ll 1$, the \ac{KL} divergence $D_{KL}(p||q)$ is proportional to $\sum_{l = 1}^L (\mu_{n,l} - \mu_{m,l})^2$.
\end{lemma}

\begin{proof}
It is easy to derive that
\begin{equation}
\begin{split}
D_{KL}(p\|q) = & \frac{1}{2}\Big[\log\frac{|\ma{\Sigma_n}|}{|\ma{\Sigma_m}|} - L +\\
& (\ve{\mu}_n-\ve{\mu}_m)^T\ma{\Sigma_m^{-1}}(\ve{\mu}_n-\ve{\mu}_m) + \\
&\operator{Tr}\left\{\ma{\Sigma_m^{-1}\ma{\Sigma}_n}\right\} \Big].
\end{split}
\label{eqn:KL_div}
\end{equation}
Because $\ma{\Sigma}_n=\ma{\Sigma}_m = \diag([\sigma^2, ...,\sigma^2])$, we have the first term in \eqref{eqn:KL_div} equals to $0$, and the last term equals to $L$. Thus, we obtain
\begin{equation}
\vspace{-2ex}
D_{KL}(p\|q) = \frac{1}{2\sigma^2}\sum_{l=1}^L(\mu_{n,l} -\mu_{m,l})^2.
\label{eqn:sim_DL}
\end{equation}
\vspace{-2ex}
\end{proof}

With Lemma \ref{lem:kl_mean_prop}, we can measure the distance between two agents more efficiently, based on the extracted $\ve{\mu}_n$ and $\ve{\mu}_m$ in the source and target latent spaces. Thus, to solve Problem \eqref{prob:TL}, we propose to choose the source agent:
\begin{equation}
i_k^\ast := \argmin_{i\in\overline{K}} \mathfrak{D}_{i,k},
    \label{eqn:agent_selection}
\end{equation}
where $\mathfrak{D}_{i,k}$ is computed based on \eqref{domain_kl} and \eqref{eqn:sim_DL}.

\begin{figure*}[!t] 
\captionsetup{justification=centering}
  \begin{minipage}[t]{0.33\textwidth}
    \centering
    \includegraphics[width=2.3in, height=1.4in]{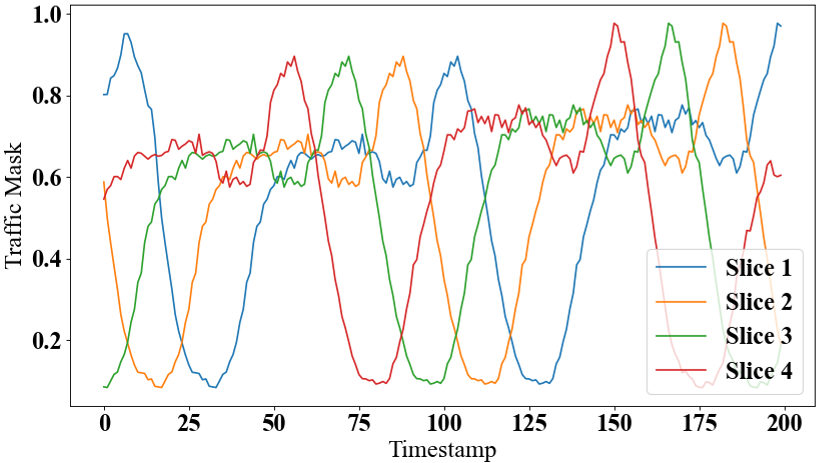}
    \captionof{figure}{\small Traffic mask to imitate the time varying network traffic}
    \label{fig:TrafMask}
  \end{minipage}
  \begin{minipage}[t]{0.33\textwidth}
    \centering
    \includegraphics[width=2.3in, height=1.4in]{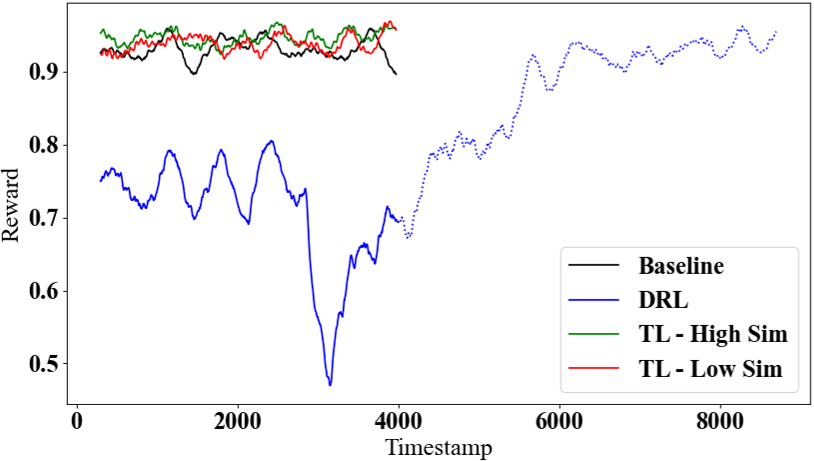}
    \caption{\small Comparing reward during the training process}
    \label{fig:DRL_TL}
  \end{minipage}
  \begin{minipage}[t]{0.33\textwidth}
    \centering
    \includegraphics[width=2.3in, height=1.4in]{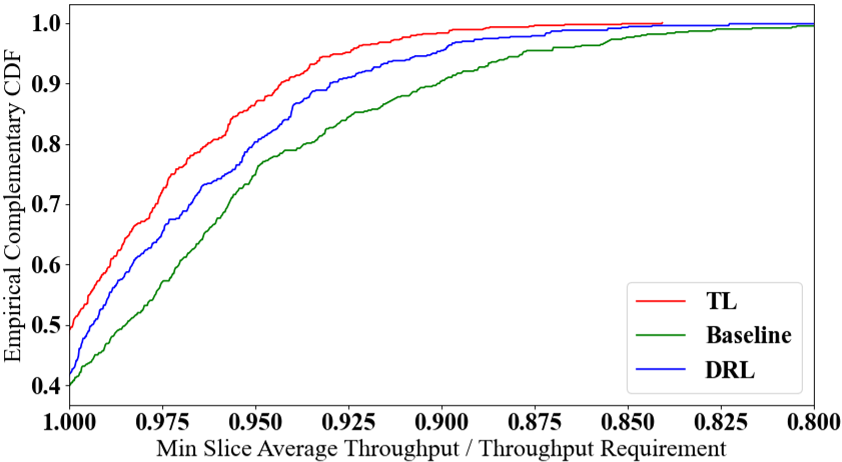}
    \captionof{figure}{\small Comparing CDF of minimum slice throughput satisfaction}
    \label{fig:CDF_Thr}
  \end{minipage}
\end{figure*}

\subsection{Integrated Transfer Learning Approach}\label{subsec:TL_Approach}

In general, the prelearned knowledge can be transferred from a source agent $i$ to the target agent $k$ with various policy transfer strategies $\Lambda(\cdot)$ and instance transfer strategy $\Gamma(\cdot)$:
\begin{itemize}
	\item {\bf Model transfer}: The policy transfer strategy $\Lambda(\cdot)$ simply initializes the target agent's policy  $\pi^{(0)}_k$ by loading the parameters (e.g., weights of the pretrained neural networks) of the pretrained policy $\pi^{(S)}_i$ from the source agent $i$. 
	\item {\bf Feature transfer}: The policy transfer strategy $\Lambda(\cdot)$ keeps partial information extracted from the source agent's pretrained policy $\pi^{(S)}_i$. In particular, the target agent loads partial of the layers (usually the lower layers) of the pretrained neural networks of $\pi^{(S)}_i$, while leaving the rest of them to be randomly initialized. Then, during training, the loaded layers are frozen and only the randomly initialized layers are fine-tuned with the instances newly collected by the target agent. 
 	\item {\bf Instance transfer}: The instance transfer strategy $\Gamma(\cdot)$ transfers the collected instances from the source agent $i$ to the target agent $k$ and saves them in the target agent's replay buffer. Then, the target agent trains a policy from scratch with randomly initialized parameters and mixed instances collected from both source and target agents.
\end{itemize}

The above-mentioned knowledge from the source domain and source task can be transferred separately or in a combined manner. In this paper, we propose the {\bf integrated transfer method} with both model and instance transfer. Specifically, the target agent $k$ initializes its local policy $\pi^{(0)}_k$ by loading the pretrained policy of the source agent $\pi^{(S)}_i$ and fine-tunes the policy by sampling from the replay buffer containing both types of instances: the instances transferred from the source agent and those locally experienced. Here, we skip the feature transfer because it practically performs well only when the similarity between the source domain/task and target domain/task is very high. Although this assumption may hold for some regression and classification tasks, we empirically find that it fails in this context of \ac{MADRL}.


\section{Performance Evaluation}
In this section, we evaluate the performance of the proposed solution within a system-level simulator \cite{SeasonII}. The simulator achieves a great accuracy in imitating the real network systems with configurable user mobility, network slicing traffic and topology. In addition, we introduce a traffic-aware baseline which allocates resource proportionally to the data traffic demand per slice. Note that the baseline assumes perfect information about per-cell per-slice traffic demands, which provides already very good results.

\subsubsection{Network settings}
We build a radio access network with $4$ three-sector sites (i.e., $K=12$ cells). All cells are deployed using LTE radio technology with $2.6$ GHz under a realistic radio propagation model Winner+ \cite{Winner}. Each cell has $N=4$ slices with diverse per-slice requirements in terms of average user throughput and delay. In the cells with label ${1, 2, 3, 7, 8, 9}$, we define per-slice average throughput requirements of $\phi^*_1=4$ MBit/s, $\phi^*_2=3$ MBit/s, $\phi^*_3=2$ MBit/s, and $\phi^*_4=1$ MBit/s respectively, and per-slice delay requirements of $d^*_1 = 3$ ms, $d^*_2 = 2$ ms, $d^*_3 = d^*_4 = 1$ ms. In the cells with label ${4, 5, 6, 10 ,11 ,12}$, we define per-slice throughput requirements as $\phi^*_1=2.5$ MBit/s, $\phi^*_2=2$ MBit/s, $\phi^*_3=1.5$ MBit/s, and $\phi^*_4=1$ MBit/s, and delay requirements of $d^*_n = 1$ ms, $\forall n\in\N$. All cells have the same radio bandwidth of $20$ MHz.

We define four groups of \ac{UE} associated to four slices in each cell respectively, each \ac{UE} group has the maximum size of $32$ and moves randomly among the defined network scenario. To mimic dynamic behavior of real user traffic, we apply a varying traffic mask $\tau_n(t)\in[0 ,1]$ to each slice to scale the total number of \ac{UE}s in each cell, Fig. \ref{fig:TrafMask} shows the traffic mask in first $200$ steps.

\subsubsection{\ac{DRL} training configuration}
For \ac{MADRL} training, we implemented \ac{TD3} algorithm at each local agent using \ac{MLP} architecture for actor-critic networks. In each \ac{TD3} model, both actor and critic neural works consist of two layers with the number of neurons as $(48, 24)$ and $(64, 24)$ respectively. The learning rates of actor and critic are $0.0005$ and $0.001$ accordingly with Adam optimizer and training batch size of $32$. We set the discount factor as $\gamma = 0.1$, since the current action has stronger impact on instant network performance than future observation. As for the training, for distributed \ac{DRL} agents we applied $3000$ steps for exploration, $5500$ steps for training, and final $250$ steps for evaluation. For \ac{TL} training process, we apply the same model setups as \ac{DRL} approaches, while only setting $4000$ steps for training and $250$ for evaluation since knowledge transfer save the time for exploration.

\subsubsection{Comparing \ac{DRL} to \ac{TL} aided approach}

\begin{figure*}[!t] 
\captionsetup{justification=centering}
  \begin{minipage}[t]{0.33\textwidth}
    \centering
    \includegraphics[width=2.3in, height=1.4in]{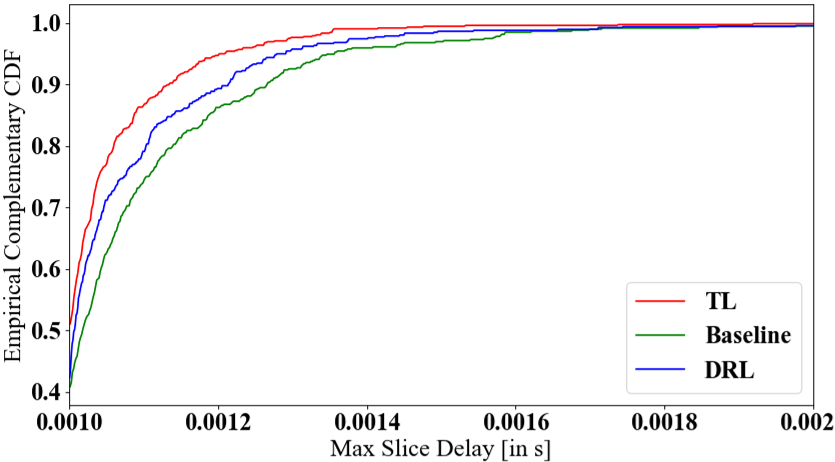}
    \captionof{figure}{\small Comparing CDF of maximum slice delay}
    \label{fig:CDF_Delay}
  \end{minipage}
  \begin{minipage}[t]{0.33\textwidth}
    \centering
    \includegraphics[width=2.3in, height=1.4in]{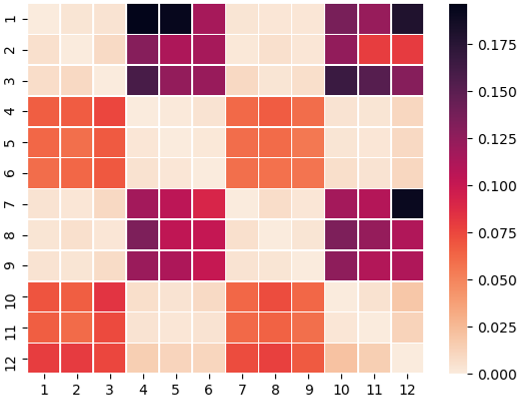}
    \caption{\small Inter-agent distance measure}
    \label{fig:Domain_Sim}
  \end{minipage}
  \begin{minipage}[t]{0.33\textwidth}
    \centering
    \includegraphics[width=2.3in, height=1.4in]{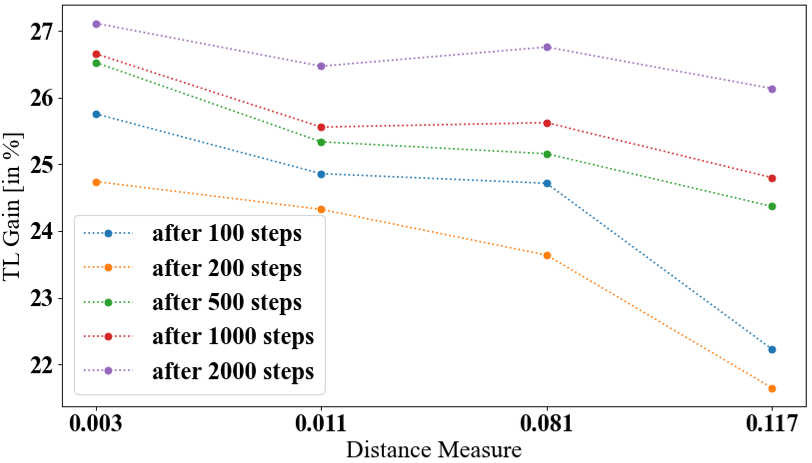}
    \captionof{figure}{\small TL performance gain depending on distance measure}
    \label{fig:TL_Compare}
  \end{minipage}
\end{figure*}

In Fig. \ref{fig:DRL_TL} we compare the evolution of reward during the training processes among the baseline, \ac{DRL} approach (proposed in Section \ref{subsec:MADRL}), and \ac{TL} approaches when transferred from source agent with low and high similarity (proposed in Section \ref{subsec:Domain_Sim} and \ref{subsec:TL_Approach}), respectively. For \ac{DRL}, we present the first $4000$ step, i.e., the same training time as \ac{TL} approaches with solid line and the rest training curve with dashed line.

As shown in Fig. \ref{fig:DRL_TL}, the distributed \ac{DRL} approach learns to achieve similar reward as baseline after a lengthy exploration phase, while both \ac{TL} approaches start with much higher start compared to \ac{DRL}. After a short fine-tuning period, the \ac{TL} approaches outperform the baseline with higher robustness, especially during the period with higher traffic demands and strong inter-cell interference where baseline has sharp performance degradation. Besides, in comparison between the \ac{TL} from agents with different similarity measure, we observe that with higher similarity, \ac{TL} provides higher start at the early stage of training, while both of them converge to similar performance after the training converges.

For performance evaluation, we compare the statistical results on minimum per slice throughput satisfaction level and maximum per slice delay, respectively, among all cells among the methods baseline, distributed \ac{DRL} and the proposed \ac{TL} approach after convergence. Fig. \ref{fig:CDF_Thr} illustrated the empirical complementary \ac{CDF} which equals $1-F_X(x)$ where $F_X(x)$ is the \ac{CDF} of minimum per slice throughput satisfaction level. We observe that the \ac{TL} approach provides the best performance comparing to others by achieving only about $12\%$ fail to satisfy $0.95$ of the requirement, while converged \ac{DRL} and baseline conclude $19\%$ and $25\%$ failure rate respectively. By average satisfaction level, the \ac{TL} approach conclude $0.92$ while \ac{DRL} and baseline only provide $0.90$ and $0.87$. Similar observation can be made from Fig. \ref{fig:CDF_Delay}, which illustrates the \ac{CDF} of maximum slice delay in ms. The \ac{TL} approach provides $1.5$ ms maximum average per-slice delay, while \ac{DRL} achieves $1.7$ ms and baseline achieves $1.8$ ms.

\subsubsection{Inter-agent similarity analysis}
We implemented the similarity analysis method introduced in Section \ref{subsec:Domain_Sim} with a \ac{VAE} model in \ac{MLP} architecture, both networks of encoder and decoder consist of 3 layers with number of neurons as $(64, 24, 4)$ and $(4, 24, 64)$ respectively. To achieve a good trade-off between low dimensional latency space and accurate reconstruction with \ac{VAE}, we map the original sample $x\in\R^{17}$ to the latent variable $\vz\in\R^{4}$.

Fig. \ref{fig:Domain_Sim} illustrates the results of inter-agent similarity analysis as a metric of distance measure proposed in \eqref{domain_kl}. It shows that our proposed method can distinguish cells with different per-slice service quality requirements and gather the cells with similar joint state-reward distribution.

\subsubsection{Dependence of \ac{TL} performance on distance measure}
In Fig. \ref{fig:TL_Compare} we compare the benefits of \ac{TL} in training process by transferring knowledge from source agents with different average inter-agent distance measures. The \ac{TL} gains are derived by comparing the reward to \ac{DRL} approach at the same training steps. The results show that before $200$ steps of \ac{TL} training, the \ac{TL} approaches with the lowest distance measure provides about $3\%$ higher gain than the one with the largest distance. As the training process continues, the gains in all \ac{TL} approaches increase with local fine-tuning and the difference between transferring from highly similar and less similar agents is getting smaller. However, \ac{TL} from the most similar agent proyvides higher gains for all training steps.

\subsubsection{Key Takeaways}: We summarized the takeaways from numerical results as follows:
\begin{itemize}
    \item All distributed \ac{DRL}-based approaches achieve better per-slice network service than the traffic-aware baseline after convergence. However, the \ac{TL} schemes outperform the conventional \ac{DRL} approach in terms of convergence rate, initial and converged performance. 
    \item Our propose \ac{VAE}-based similarity measure well quantifies the distance between agents and can be used to suggest a mapping from the defined distance measure to the transfer learning performance gain. 
    \item The difference between the gains achieved by \ac{TL} from the highly similar and the less similar agents is more significant when the number of training steps is low (i.e., with limited online training samples). Although the advantage of transferring from a highly similar agent over a less similar agent decreases when the number of online training steps increases, a slight performance gain is always achieved by transferring knowledge from the most similar source agent.
\end{itemize}

\section{Conclusion}

In this paper, we formulated the dynamic inter-slice resource partitioning problem to optimize the network requirement satisfaction level of all slices in each cell. To tackle the inter-cell interference, we proposed a coordinated \ac{MADRL} method with the coordination scheme of information sharing. We proposed a novel integrated \ac{TL} method to transfer the learned DRL policies among different local agents for accelerating the policy deployment. The method is accomplished by a new inter-agent similarity measurement approach and a new knowledge transfer approach. We evaluated the proposed solution with extensive simulations in a system-level simulator, where the results show our approach outperforms conventional \ac{DRL} solutions.

\section*{Acknowledgment}
This work was supported by the German Federal Ministry of Education and Research (BMBF) project KICK [16KIS1102K].

\bibliographystyle{IEEEtran}
\bibliography{myreferences}

\end{document}